\documentclass[a4paper,UKenglish]{lipics-v2016}

\usepackage{anyfontsize}
\usepackage{microtype}
\usepackage{mathtools}
\usepackage{amsfonts}
\usepackage[usenames,dvipsnames]{xcolor}
\usepackage[toc,page]{appendix}
\usepackage{tikz}
\usepackage{esvect}
\usepackage{xcolor}
\usepackage{hyperref}
\usepackage[ruled]{algorithm2e}
\usepackage{thmtools,thm-restate}

\makeatletter
\renewcommand{\@algocf@capt@plain}{above}
\makeatother

\usetikzlibrary{patterns}
\usetikzlibrary{decorations.pathreplacing}

\theoremstyle{definition}

\newcommand{\floor}[1]{\left\lfloor #1 \right\rfloor}

\newcommand{\dotpr}[2]{\left\langle #1 , #2 \right\rangle}

\newcommand{\dist}[1]{\lVert #1 \rVert}
\newcommand{\prob}[1]{\mathbb{P}\left[ #1 \right]}

\newcommand{\CNN}[2]{\textsf{NN}_{\mathrm{wfn}}(#1, #2)}
\newcommand{\CNNE}{$\textsf{NN}_{\mathrm{wfn}}$}
\newcommand{\pre}[2]{preproc(#1, #2)}
\newcommand{\que}[2]{query(#1, #2)}
\newcommand{\Oc}{\mathcal{O}}
\newcommand{\Rspace}{\mathbb{R}}
\newcommand{\Rdspace}{\mathbb{R}^d}
\newcommand{\Rkspace}{\mathbb{R}^k}

\newcommand*{\CNNT}{$c$--approximate near neighbor}

\title{Improved approximate near neighbor search without false negatives for $l_2$}

\author{Piotr Wygocki}
\affil{University of Warsaw, Poland\\
  \texttt{wygos@mimuw.edu.pl}}

\authorrunning{P. Wygocki} 

\Copyright{Piotr Wygocki}

\subjclass{F.2.2; G.3}
\keywords{approximate near neighbor search, high-dimensional, similarity search, locality sensitive hashing}%

\begin{document}

\maketitle

\begin{abstract}

We present a new algorithm for the $c$--approximate nearest neighbor search
without false negatives for $l_2^d$. 
We enhance the dimension reduction method presented in \cite{wygos_red} and combine it with the standard results of Indyk and Motwani~\cite{motwani}.
We present an efficient algorithm with Las Vegas guaranties for any $c>1$. 
This improves over the previous results, which require $c=\omega(\log\log{n})$ \cite{wygos_red}, where $n$ is the number of the input points.
Moreover, we improve both the query time and the pre-processing time.

Our algorithm is tunable, which allows for different compromises between the query and the pre-processing times.
In order to illustrate this flexibility, we present two variants of the algorithm. The ''efficient query'' variant involves the query time of $\Oc(d^2)$ and the polynomial pre-processing time.
The ''efficient pre-processing'' variant involves the pre-processing time equal to $\Oc(d^{\omega-1} n)$ and the query time sub-linear in $n$, where $\omega$
is the exponent in the complexity of the fast matrix multiplication.

In addition, we introduce batch versions of the mentioned algorithms, where the queries come in batches of size $d$. 
In this case, the amortized query time of the ''efficient query'' algorithm is reduced to $\Oc(d^{\omega -1})$.

\end{abstract}

\section{Introduction}

The near neighbor problem has various applications in image processing, search engines, recommendation engines, prediction and machine learning.
We define the near neighbor problem as follows: for a given input set, a query point and distance $ R $, 
return a point (optionally all points) from the input set closer  than $ R $ to the query point in some metric (usually $ l_p $ for $ p \in [1, \infty] $) 
or report that such a point does not exist.\footnote{Some authors refer to this problem as the Point Location in Equal Balls (PLEB) \cite{motwani}.}
The input set and the distance $ R $ are known in advance. Because of this, the input set can be preprocessed, which can afterwards shorten the query time.
The problem of finding the nearest neighbor with no $R$ given, can be efficiently reduced to the problem defined as above \cite{motwani}.

Unfortunately, the near neighbor search is hard for high dimensional spaces such as $ l_p^d $ for a large~$d$.
The existence of an algorithm with the query time sub-linear in the input set size and non-exponential in $d$ and the pre-processing time non-exponential in $ d $ 
would contradict the strong hypothetical time hypothesis~\cite {Williams2005}.
In order to overcome this obstacle, the $c$--near neighbor problem was introduced.
In this problem, the query result is allowed to contain points located at a maximum distance of $ cR $ from the query point.
In other words, the points located at a distance smaller than $R$ from the query point are classified as neighbors, points further than $ cR $ are classified as ''far away'' points,
while the rest can be classified in either of these two categories. Naturally, we consider $c> 1$.
This assumption makes the problem easier for many metric spaces, such as $ l_p $ for $ p \in [1,2] $ or the Hamming \cite {motwani} space.
On the one hand, the queries are sub-linear in the input size. On the other hand, queries and pre-production time are
polynomial in the dimension of space.

Many previously known algorithms for the \CNNT{} 
use locality sensitive hashing and give Monte Carlo guarantees
for the returned points (see, for example, \cite {DBLP:journals/cacm/AndoniI08, CHAZELLE200824, motwani}). 
That is, any input point within the distance $R$ from the query point is classified as neighbor with some probability, which means there may be false negatives.
Locality sensitive hashing functions are functions which roughly preserve distances, i.e., given two points, the distance between their hashes approximates
the distance between them with high probability.
A common choice for the hash functions is $f(x) = \dotpr{x}{v}$ or $f(x) = \floor{\dotpr{x}{v}}$,
where $v$ is a vector of numbers drawn independently from some probability distribution \cite{DBLP:journals/cacm/AndoniI08, motwani, wygos}.
For the Gaussian distribution, $\dotpr{x}{v}$ is also Gaussian with zero mean and standard deviation equal to $\dist{x}_2$.
It is easy to see that these are LSH functions for $l_2$, but as mentioned above, they only provide probabilistic guaranties.
In this paper, we aim to enhance this by focusing on the \CNNT{} search without false negatives for $l_2$.
In other words, we consider algorithms, where a point 'close' to the query point is guaranteed to be returned. Such class of guaranties is
often called Las Vegas.
An algorithm with Las Vegas guaranties can be adjusted to one with Monte Carlo guaranties. 
Markov's inequality implies, that if the expected value of the computation time is small, then with large probability the computation time is also small.
We can break the computation after the certain amount of time passed and return the empty result which gives Monte Carlo guaranties.

Throughout this paper, we assume that the size of the input set  $n \gg d$ and $\exp(d) \gg n$. This represents a situation where the exhaustive scan through all the input points, as well as the
usage of data structures exponentially dependent on $d$, become intractable. If not explicitly specified, all statements assume the usage of the $l_2$ norm.

\section{Related Work}

\subsection{Algorithms for constant dimension}
There is a number of algorithms for the \CNNT{} problem assuming constant $d$ \cite{Arya:1995:OAA:241626,Kleinberg:1997:TAN:258533.258653, Clarkson:1994:AAC:177424.177609}. 
In each of them, either the pre-processing time or the query time depends exponentially on $d$. 
Nevertheless, these are the best fully deterministic algorithms that are known \cite{Arya:1995:OAA:241626, motwani}.
A particularly interesting algorithm is presented in \cite{motwani}, having the pre-processing time
 $n\Oc(1/\epsilon)^d$ and the query time equal to $\Oc(d)$, where $\epsilon=c-1$. We will use this algorithm to obtain our results.

\subsection{Monte Carlo algorithms}
There exists an efficient Monte Carlo \CNNT{} algorithm for $l_1$ with the query 
and the pre-processing complexity equal to $\Oc(n^{1/c})$ and $\Oc(n^{1+1/c})$, respectively \cite{motwani}.
For $l_2$ in turn, there exists a near to optimal algorithm 
with the query and the pre-processing complexity equal to $n^{1/c^2+ o(1)}$ and $n^{1+1/c^2 + o(1)}$, respectively \cite{DBLP:journals/cacm/AndoniI08} \cite{O'Donnell:2014:OLB:2600088.2578221}.
Moreover, the algorithms presented in \cite{motwani} work for $l_p$ for any $p\in [1,2]$.
There are also data dependent algorithms, taking into account the actual distribution of the input set \cite{data-depended-hashing},
which achieve query time  $dn^{\rho+o(1)}$ and space $\Oc(n^{1+\rho+o(1)}+dn)$, where $\rho=1/(2c^2-1)$.

Recently, the optimal hashing--based time--space trade-offs for the \CNNT{} in $l_2$ were considered \cite{Andoni:2017:OHT:3039686.3039690}.
For any $p_u,p_q\ge 0$ such that:
\[c^2\sqrt{p_q} + (c^2-1)\sqrt{p_u} = \sqrt{2c^2-1}, \]
there is a \CNNT{} algorithm with the storage $\Oc(n^{1+p_u+o(1)} + dn)$ and the query time $n^{p_q + o(1)} + dn^{o(1)}$.


\subsection{Las Vegas algorithms}
Pagh~\cite{Pagh15} considered the \CNNT{} search without false negatives~(\CNNE{}) for the Hamming space, obtaining the results close to those presented in \cite{motwani}.
He showed that the bounds of his algorithm for $cR = \log( n/k)$ differ by at most a factor of $\ln 4$ in
the exponent in comparison to the bounds in \cite{motwani}. 
Recently, Ahle showed an optimal algorithm for the \CNNT{} without false negatives for the Hamming space
and  Braun-Blanquet metric \cite{DBLP:journals/corr/Ahle17}\cite{O'Donnell:2014:OLB:2600088.2578221}.
Indyk~\cite{DBLP:conf/focs/Indyk98} provided a deterministic algorithm
for $l_{\infty}$ for $c = \Theta(\log_{1+\rho}\log{d})$ with the storage $\Oc(n^{1+\rho}\log^{O(1)}n)$ and the query time  $\log^{O(1)}n$ for some tunable parameter $\rho$. He proved
that the  \CNNT{} without false negatives for  $l_{\infty}$ for $c<3$ is as hard as the subset query problem, a long-standing combinatorial problem. 
This indicates that the  \CNNT{} without false negatives for  $l_{\infty}$ might be hard to solve
for any $c>1$.

Indyk~\cite{Indyk:2007:UPE:1250790.1250881} considered deterministic mappings $l_2^n \rightarrow l_1^m$, for $m=n^{1+o(1)}$, which might be useful for constructing
efficient algorithms for the  \CNNT{} without false negatives. 
If we were able to efficiently embed $l_1$ into the Hamming space (which is just $\{0,1\}^d$ with $l_1$ distance function) with additional guaranties for false negatives,
it would also give an algorithm for $l_2$ and $l_1$.\footnote{
To the best of the author's knowledge, such an embedding will be presented at FOCS 2017 in the conference version of the paper of Ahle~\cite{DBLP:journals/corr/Ahle17}.}

Algorithms for any $p \in [1,\infty]$ are presented in \cite{wygos}.
Two hashing function families are considered, giving different trade-offs between the execution times and the conditions on $c$.
Unfortunately, these algorithms work only for $c>\sqrt{d}$. In further work, Sankowski~et~al.~\cite{wygos_red} showed a dimension reduction 
technique with Las Vegas guaranties. Application of the algorithm introduced in \cite{wygos} to the problem with the reduced dimension results in an algorithm for $c=\Omega(\sqrt{\log{n}})$.
This might be further reduced to $c=\omega(\sqrt{\log\log{n}})$ \cite{wygos_red}. 

In this work, we use the dimension reduction introduced in \cite{wygos_red} and apply the algorithm of Indyk and Motwani~\cite{motwani} to the problem in the reduced space. 
After a slight strengthening of the results of \cite{wygos_red}, we get the algorithms for any $c$.

\subsubsection{Las Vegas dimension reduction}\label{dim_red}

 Sankowski~et~al.~\cite{wygos_red} showed that:
 \begin{restatable}{lem}{primelemma}[Reduction Lemma -- Lemma 1 in \cite{wygos_red}]\label{rjl}

For any parameter $\alpha \ge 1$ and $k < d$, there exist $d/k$ 
linear mappings $A^{(1)}, A^{(2)}, \dots, A^{(d/k)}$, from $\Rdspace$ to $\Rkspace$, such that:
\begin{enumerate}
    \item for each point $x\in \Rdspace$ such that $\dist{x}_2 \le 1$, there exists $1 \le i \le d/k$, which satisfies $\dist{A^{(i)} x}_2 \le 1$,

    \item for each point $x \in \Rdspace$ such that $\dist{x}_2 \ge c$, where $c>1$,
        for each $i$: $1 \le i \le d/k$, we have
        \[\prob{\dist{A^{(i)}x}_2 \le \alpha} < e^{\frac{k}{2}\big(1 - (\frac{\alpha}{c})^2 + \log((\frac{\alpha}{c})^2)\big)} \] 

\end{enumerate}

\end{restatable}
Since $\log{x} < x - 1 - (x - 1)^2 /2$ for $x<1$, the above bound is not trivial for $\alpha < c$.
Applying the reduction lemma gives the following reduction for the  \CNNT{} without false negatives:

 \begin{restatable}{cor}{primecor}[generalization of Corollary 3 in \cite{wygos_red}]\label{colr}

    For any $1 \le \alpha < c$ and $\gamma\log{n} < d$, the $\CNN{c}{d}$ can be
    reduced to $\Oc(d/(\gamma\log{n}))$ instances of the $\CNN{\alpha}{\gamma\log{n}}$, where $\gamma <\frac{2(1-\nu)}{(\frac{\alpha}{c})^2 -1 -2\log\frac{\alpha}{c}} $ 
    for a tunable parameter $\nu \in[0,1)$
    and:

    \[\que{c}{d} = \Oc(d^2 + n^{\nu} +  d/(\gamma\log{n})\ \que{\alpha}{\gamma\log{n}}),\]
   \[\pre{c}{d} =\Oc(d^{\omega-1} n  + d/(\gamma\log{n})\ \pre{\alpha}{\gamma\log{n}}).\]
   If the queries are provided in the batches of the size $d$, we obtain the algorithm with the query time:
   \[\que{c}{d} = \Oc(d^{\omega-1} + n^{\nu} +  d/(\gamma\log{n})\ \que{\alpha}{\gamma\log{n}}).\]
\end{restatable}

In the above version of Corollary \ref{colr}, we generalize the Corollary 3 from \cite{wygos_red} in the following way:
\begin{itemize}
\item We introduce an additional parameter $\nu$, which allows us to set different compromises between the pre-processing time and the query time. 
\item We observe, that given the assumption of $n \gg d$, the preprocessing of all points can be expressed as the multiplication of matrices of dimensions $n \times d$ and $d \times d$ which can be performed
in time $\Oc(nd^{\omega -1})$. Analogical argument leads to the conclusion that the query time of the batch version is $\Oc(d^{\omega -1})$.
In further theorems, we provide the amortized  query times for the batch version. This can be easily turned into a non-batch version
by substituting the $d^{\omega-1}$ term with $d^2$ in the query complexities.
\item We present slightly stronger bound for $\gamma$.
\end{itemize}

The proof of Corollary \ref{colr} is essentially the same as the proof presented in \cite{wygos_red}.
For the reader's convenience, this proof is included in Appendix \ref{dim_red_app}.

Combining the above corollary with the results introduced in \cite{wygos}, we can achieve
an algorithm with the polynomial pre-processing time and the sub-linear query time for $c=\Theta(\sqrt{\log{n}})$ \cite{wygos}. In this paper, we relax this restriction and show an algorithm for any $c>1$.

\section{Our contribution}
Recently, efficient algorithms were proposed for solving the \CNNT{} search without false negatives
for $c = \Omega( \max\{\sqrt{d}, d^{1-1/p}\})$ in $l_p$ for any $p \in [1,\infty]$ and for $c=\omega(\sqrt{\log\log{n}})$ for $l_2$ \cite{wygos, wygos_red}.
The main problem with these algorithms is the constraint on~$c$. 
The contribution of this paper is relaxing this condition and improving the complexity of the algorithms for $l_2$:

\begin{theorem}\label{main}
The $\CNN{c}{d}$ can be solved with the amortized query time  $\Oc(d^{\omega-1} + n^{\nu})$ and the pre-processing time:
\begin{itemize}
 \item  $n^{1+\Oc(\frac{1-\nu}{\epsilon^2}\log{\frac{1}{\epsilon}})}$ for any $c<2$,  
  \item $\Oc(d^{\omega -1}n+dn^{1+\Oc(\frac{1-\nu}{\log{c}})}/\log{n})$ for any $c\ge2$,
\end{itemize}
for some tunable parameter $\nu \in [0,1)$ and $\epsilon=c-1$.
\end{theorem}
As mentioned before, we assume, that the queries are provided in batches of size $d$.  This assumption can be omitted, which leads to an algorithm with the query time of $\Oc(d^{2} + n^{\nu})$.
The above is also valid for the other presented algorithms (in particular, for Theorem \ref{main_pre}). We focus on the batch version to avoid unnecessary complexity.

In particular, for $\nu = 0$ and $c\le2$, we achieve an algorithm with the query time $\Oc(d^{\omega-1})$ and the pre-processing time $n^{\Oc(\frac{1}{\epsilon^2}\log{\frac{1}{\epsilon}})}$.
For small $c$, our results are  incomparable with the previously discussed algorithms. 
In particular, our results are similar to the algorithm presented recently in \cite{Andoni:2017:OHT:3039686.3039690}, 
which works with the query time $n^{o(1)}$ and the pre-processing time $n^{\Oc(1/\epsilon^2) + o(1)}$.
Results presented in \cite{DBLP:journals/cacm/AndoniI08, Andoni:2017:OHT:3039686.3039690} give weaker Monte Carlo guaranties.
Increasing the parameter $\nu$ allows us to reduce the preprocessing complexity. 
For $\nu = 1/c$ we achieve an algorithm with the query time $\Oc(n^{1/c})$ and the pre-processing time $n^{\Oc(\frac{1}{\epsilon}\log{\frac{1}{\epsilon}})}$.
Setting $\nu = 1-\epsilon^2 /\log{\frac{1}{\epsilon}}$ gives the algorithm with the polynomial pre-processing time independent of $c$.

In addition, we show the pre-processing efficient versions of the algorithms, which have an optimal in therms of $n$, linear complexity.
\begin{theorem}\label{main_pre}
The $\CNN{c}{d}$ can be solved with the pre-processing time  $\Oc(d^{\omega-1} n)$ and the amortized query time:
\begin{itemize}
 \item  $\Oc(d^{\omega-1} + d n^{\frac{1}{1+\Oc(\epsilon^2 \log^{-1}{\frac{1}{\epsilon}})}})$ for any $c<2$,  
  \item $\Oc(d^{\omega-1} + dn^{\Oc(\frac{1}{\log{c}})})$ for any $c\ge2$,
\end{itemize}
where  $\epsilon=c-1$.

\end{theorem}
This gives new results for probably the most interesting case from the practical point of view.
In particular, for $c\le2$, we achieve the query time:
\[\Oc(d^{\omega-1} + d n^{\frac{1}{1+\Oc(\epsilon^2 \log^{-1}{\frac{1}{\epsilon}})}})
= \Oc(d^{\omega-1} + d n^{1-\Oc(\epsilon^2 \log^{-1}{\frac{1}{\epsilon}})})
 = \Oc(d^{\omega-1} + d n^{1-\Oc(\epsilon^3)})\]
Again, our algorithm gives results similar to one presented in \cite{Andoni:2017:OHT:3039686.3039690} which gives the query time equal to $\Oc(d + n^{1- \Oc(\epsilon^2) +o(1)})$.

All of the presented algorithms give Las Vegas guaranties, which are stronger than the previously considered Monte Carlo guaranties. 
The provided algorithms are practical in terms of implementation.

\begin{center}
 \begin{tabular}{|c|  c | c| c| c|} 
 \hline
 author & guaranties  & query& pre-processing & space  \\ [0.5ex] 
 \hline
 \cite{DBLP:journals/corr/Ahle17} & Las Vegas  & $\Oc(n^{1/c})$& $\Oc(n^{1+1/c})$ & Hamming for $c>1$\\ 
 \hline
 \cite{DBLP:conf/focs/Indyk98} & Deterministic  & $\log^{O(1)}n$  &  $\Oc(n^{1+\rho}\log^{O(1)}n)$ & $l_{\infty}$ for $c=\Theta(\log_{1+\rho}\log{d})$  \\
 \hline
 \cite{wygos_red} & Las Vegas  & $ \tilde{\Oc}(d^2 + d n^{o(1)})$ &  $\tilde{\Oc}(d^2 n  + d n^{1 + \frac{\ln{3}}{\ln(c/\mu)}+ o(1)})$ & $l_2$ for $c>\mu=\omega(\sqrt{\log\log{n}})$ \\
 \hline
 \cite{Andoni:2017:OHT:3039686.3039690}  & Monte Carlo  & $\Oc(d + n^{1- \Oc(\epsilon^2) +o(1)})$ & $\Oc(dn + n^{1 +o(1)})$ & $l_2$ for $c>1$ \\
 \hline
  \textbf{this work} & Las Vegas  & $\Oc(d^{\omega-1} + d n^{1-\Oc(\epsilon^2 \log^{-1}{\frac{1}{\epsilon}})})$ & $\Oc(d^{\omega-1}n)$ & $l_2$ for $c>1$ \\ [1ex] 
  \hline
 \cite{Andoni:2017:OHT:3039686.3039690}  & Monte Carlo & $n^{o(1)}$ & $n^{\Oc(1/\epsilon^2) +o(1)}$ & $l_2$ for $c>1$ \\
 \hline
 \textbf{this work} & Las Vegas  & $\Oc(d^{\omega-1})$& $n^{\Oc(1/\epsilon^2\log{1/\epsilon})}$ & $l_2$ for $c>1$ \\ [1ex] 
 \hline
 
\end{tabular}
\captionof{table}{Comparison of the results for the \CNNT{}. We present only ``fast query'' and ``fast pre-processing'' parts of results for possibly small $c$. 
Also, results presented in  \cite{Andoni:2017:OHT:3039686.3039690} are under  assumption that $d=n^{o(1)}$). Results in \cite{DBLP:conf/focs/Indyk98} are for a tunable parameter $\rho$.}
\end{center}

\section{Notations}

The input set is always assumed to contain $n$ points. 
The \CNNT{} search without false negatives with parameter $c > 1$ and the dimension of the space equal to $d$, is denoted as
$\CNN{c}{d}$.  The considered problem is solved in the standard euclidean space (i.e., the standard norm in $l_2$)
$\dist{\cdot}_2$: $\dist{x}_2 = (\sum_i{|x_i|^2})^{1/2}$. The expected query and pre-processing time complexities of
$\CNN{c}{d}$ will be denoted as $\que{c}{d}$ and $\pre{c}{d}$ respectively.

In this paper, we use the term “pre-processing” to refer to the sum of
the actual pre-processing time and the storage required. 
W.l.o.g, throughout this work we assume, that $R$ -- a given radius
equals 1 (otherwise, all vectors' lengths might be rescaled by $1/R$). 
In this work, it is often convenient to use $\epsilon = c-1$ instead of $c$.
Whenever $\epsilon$ appears, it is assumed to be defined as above.
Finally, we use $\omega$
to denote the exponent in the complexity of the fast matrix multiplication (currently $\omega \approx 2.37$).

\section{Nearest neighbors without false negatives for any  \texorpdfstring{$\boldsymbol{c}$}{Lg}} \label{sec_qp}

In this section, we show an efficient algorithm for solving the $\CNN{c}{d}$ in $l_2$.
Indyk and Motwani~\cite{motwani} showed an algorithm with the pre-processing time
 $n\Oc(1/\epsilon)^d$ and the query time equal to $\Oc(d)$, where $\epsilon=c-1$.
The idea of the algorithm is the following.
We start with a quantization of the given space, which reduces the problem to finding the near neighbor in a space with integer coefficients. 
After the quantization, there is a finite number of points which have a neighbor in the input set.
It is enough to provide the data structure which will store all such points with accompanying near neighbors from the input set. 
It is  proved that the number of neighbors of each input point is $\Oc(1/\epsilon)^d$. 
So in total, we need to store $n\Oc(1/\epsilon)^d$ of such points. 
We can fetch a point from the data structure in the time proportional to the size of this point, thus the query time is $\Oc(d)$.

The only issue left is to provide appropriate storage. Indyk and Motwani~\cite{motwani} provided a hash-map storage.
Let us consider the following standard, deterministic construction. For each of the stored points, we store its 
 bit representation in a binary tree. This way the length of the branch representing a point equals to its bit-length. 
 Hence, the query time is proportional to the bit size of the query. 
 The size of the whole tree is bounded by the total size of the binary representation of all the stored points.
 
The above construction gives an algorithm for the \CNNT{} in $l_2$ with the efficient query time. 
Unfortunately, unless $d=\Oc(\log{n})$, the pre-processing time is exponential. 
If the dimension is larger than $\gamma \log{n}$, with $\gamma$ defined as in Corollary \ref{colr}, we may reduce the complexity of the pre-processing by reducing the
dimension of the input space.  

\subsection{Fast query}\label{sec_q}
In this section we prove  Theorem \ref{main}:
\begin{itemize}
 \item $c<2$:
 
 For $c<2$, we set $\alpha = \frac{c+1}{2}$ in Corollary \ref{colr}. 
  It follows that $\CNN{c}{d}$ can be reduced to $\Oc(d/(\gamma\log{n}))$ instances of $\CNN{ c/2+1/2}{\gamma\log{n}}$ and:
    \begin{itemize}
     \item the query time equals $\Oc\big(d^{\omega-1} + n^{\nu} +  d/(\gamma\log{n})\ \que{\frac{c+1}{2}}{\gamma\log{n}}\big)$,
    
    \item the pre-processing time equals $\Oc\big(d^{\omega-1} n  + d/(\gamma\log{n})\ \pre{\frac{c+1}{2}}{\gamma\log{n}}\big).$ 
    \end{itemize}
    
    Since  $\log{x} < x - 1 - (x - 1)^2 /2$ for $x<1$, 
    \[\gamma<\frac{2(1-\nu)}{(\frac{\alpha}{c})^2 -1 -2\log\frac{\alpha}{c}} < (1-\nu)\Big(\frac{2c^2}{c^2-\alpha^2}\Big)^2=\Oc\Big(\frac{1-\nu}{\epsilon^2}\Big).\]
 
 Consequently, we reduce the problem to $\Oc(d/(\gamma\log{n}))$ instances of $\CNN{\frac{c+1}{2}}{\Oc(\frac{1-\nu}{\epsilon^2}\log{n})}$.
 By~\cite{motwani}, each of these instances is solved with the pre-processing time $\Oc\big(\frac{1}{\epsilon}\big)^{\frac{1-\nu}{\epsilon^2}\log{n}} = n^{\Oc(\frac{1-\nu}{\epsilon^2}\log{\frac{1}{\epsilon}})}$ 
 and with the query time  $\Oc(\gamma\log{n})$.
 \item $c\ge2$:

 After setting $\alpha$ to any constant value such that $\log{2} > \log{\alpha}  + 1/2$ in Corollary \ref{colr}, we reduce the problem to $\Oc(d/(\gamma\log{n}))$ instances of $\CNN{\Oc(1)}{\Oc(\gamma\log{n})}$.
 We have:
    \[\gamma<\frac{2(1-\nu)}{(\frac{\alpha}{c})^2 -1 -2\log\frac{\alpha}{c}} <\frac{1-\nu}{\log{c} - \log{\alpha} -\frac{1}{2}} = \Oc\Big(\frac{1-\nu}{\log{c}}\Big).\]
 Each of these instances is solved with the pre-processing time $n\Oc(1)^{\frac{(1-\nu)\log{n}}{\log{c}}} =  n^{1+\Oc(\frac{1-\nu}{\log{c}})}$ and with the query time  $\Oc(\gamma\log{n})$. 
 
\end{itemize}  
This ends the proof of the Theorem \ref{main}.

\subsection{Fast pre-processing}
In this section, we prove Theorem \ref{main_pre}.
To achieve the algorithm with the fast, linear pre-processing, we will store all input points in the hash-map. 
During the query phase, we will ask the hash-map for all of $\Oc(1/\epsilon)^d$ of points which are close to this query point.
This way, most of the computation is moved from the pre-processing to the query phase. The dimension reduction is used in a similar manner as in Section \ref{sec_q}.
We skip the computation steps which are analogical to the corresponding ones in the previous section.

\begin{itemize}
 
\item For $c<2$, after setting $\alpha = c/2+1/2$ in Corollary \ref{colr} we get an algorithm with:
\begin{itemize}
     \item 
     the query time: $\Oc\big(d^{\omega-1} + n^{\nu} +  d/(\gamma\log{n})\ \que{\frac{c+1}{2}}{\gamma\log{n}}\big) = \Oc\big(n^{\nu} +   n^{\Oc(\frac{1-\nu}{\epsilon^2}\log{\frac{1}{\epsilon}})}\big)$,
    
    \item the pre-processing time: $\Oc\big(d^{\omega-1} n  + d/(\gamma\log{n})\ \pre{\frac{c+1}{2}}{\gamma\log{n}}\big) = \Oc(d^{\omega-1} n).$
\end{itemize}
  Let us assume that the query time is $\Oc\big(n^{\nu} +   n^{D\frac{1-\nu}{\epsilon^2}\log{\frac{1}{\epsilon}}}\big)$. After setting
  \[\nu  = \frac{D\frac{1}{\epsilon^2}\log{\frac{1}{\epsilon}}}{D\frac{1}{\epsilon^2}\log{\frac{1}{\epsilon}} + 1},\]
  we get the query time complexity equal to $\Oc(n^{\frac{1}{\epsilon^2 D^{-1} \log^{-1}{\frac{1}{\epsilon}} + 1}})$.

 \item For $c\ge2$, after setting $\alpha$ to any constant value such that $\log{2} > \log{\alpha}  + 1/2$ and $\nu=0$ in Corollary \ref{colr}, we get the algorithm with:
  \begin{itemize}
  \item 
     the query time: $\Oc\big(d^{\omega-1} + n^{\nu} +  d/(\gamma\log{n})\ \que{\alpha}{\gamma\log{n}}\big) = \Oc(dn^{\Oc(\frac{1}{\log{c}})})$
    
    \item the pre-processing time: $\Oc\big(d^{\omega-1} n  + d/(\gamma\log{n})\ \pre{\alpha}{\gamma\log{n}}\big) = \Oc(d^{\omega-1} n).$
  \end{itemize}  
\end{itemize}
This ends the proof of the Theorem \ref{main_pre}.

One can produce the Monte Carlo version of Theorems \ref{main} and \ref{main_pre}, which have only slightly better complexities (some factors of $d$ would be removed), 
 because the dimension reduction is simpler in this case.

\section{Conclusion and Future Work}

We have presented the  \CNNT{} algorithm without false
negatives in $l_2$ for  any $c$. 
The future works concerns reducing the time complexity of the algorithm
or proving that these
restrictions are essential. We wish to match the time complexities given in
\cite{motwani} or show that the achieved bounds are optimal.


\bibliography{bib}

\appendix
\section{Dimension Reduction}\label{dim_red_app}
In this section, we repeat arguments presented in \cite{wygos_red}.
Let us start with the well-known Johnson-Lindenstrauss Lemma which is crucial to our results:
\begin{lemma}[Johnson-Lindenstrauss]\label{jl1}
    Let $Y\in \Rspace^d$ be  chosen  uniformly  from  the  surface  of  the $d$-dimensional  sphere.   Let
    $Z=(Y_1,Y_2,\dots,Y_k)$ be  the  projection  onto  the  first $k$ coordinates,  where $k < d$.   Then  for  any
    $\alpha < 1$
    :
    \begin{equation}
    \prob{\frac{d}{k}\dist{Z}_2^2 \le \alpha} \le \exp(\frac{k}{2}(1-\alpha+\log{\alpha})),
    \end{equation}

\end{lemma}

The basic idea is the following: we will introduce a number of linear mappings
to transform the $d$-dimensional problem to a number of problems with
 reduced dimension.

We will introduce $d/k$\footnote{For simplicity, let us assume that $k$ divides $d$, 
this can be achieved by padding extra dimensions with $0$'s.}
linear mappings $A^{(1)}, A^{(2)}, \dots, A^{(d/k)}:R^d \rightarrow R^k$, where $k<d$ and
show the following properties:
\begin{enumerate}
    \item for each point $x \in \Rdspace$, such that $\dist{x}_2 \le 1$, there exists $1 \le i \le d/k$, such that $\dist{A^{(i)} x}_2 \le 1$,
    \item for each point $x \in \Rdspace$, such that $\dist{x}_2 \ge c$, where $c>1$, the probability that there exists $1 \le i \le d/k$, such that $\dist{A^{(i)}x}_2 \le 1$ is bounded.
\end{enumerate}

The property $1.$ states, that for a given 'short' vector (with a length
smaller than $1$), there is always at least one mapping, which transforms this
vector to a vector of length smaller than $1$.  Moreover, we will show, that
there exists at least one mapping $A^{(i)}$, which does not increase the length
of the vector, i.e., such that $\dist{A^{(i)} x}_2 \le \dist{x}_2$. The
property $2.$ states, that we can bound the probability of a 'long' vector
($\dist{x}_2 > c$), being mapped to a 'short' one ($\dist{A^{(i)} x}_2 \le
1$). Using the standard concentration measure arguments, we will prove that
this probability decays exponentially in $k$.

\subsection{Linear mappings}

In this section, we will introduce linear mappings satisfying properties 1. and 2.
Our technique will depend on the concentration bound used to
prove the classic Johnson-Lindenstrauss Lemma.
In Lemma \ref{jl1},  we take a random vector and project it to
the first $k$ vectors of the standard basis of $\Rspace^d$.  In our settings,
we will project the given vector to a random orthonormal basis which gives the
same guaranties.
The mapping  $A^{(i)}$ consists
of $k$ consecutive vectors from the random basis of the
$\Rspace^d$ space scaled by $\sqrt{\frac{d}{k}}$.
The following reduction describes the basic properties of
our construction (presented also  in related work):
\primelemma*
\begin{proof}
Let $a_1, a_2, \dots, a_d$ be a random basis of $R^d$. Each of the
$A^{(i)}$ mappings is represented by a $k \times d$ dimensional matrix.  We will use $A^{(i)}$ for denoting both the mapping and the corresponding matrix.
The $j$th row of the matrix $A^{(i)}$ equals $A^{(i)}_j = \sqrt{\frac{d}{k}}
a_{(i-1)k+j}$. In other words, the rows of $A^{(i)}$ consist
of $k$ consecutive vectors from the random basis of the
$\Rspace^d$ space scaled by $\sqrt{\frac{d}{k}}$.

To prove the first property, observe that $A = \sum_{i=1}^d \dotpr{a_i}{x}^2 \le
1$, since the distance is independent of the basis.  Assume on the contrary, that for each $i$,
$\dist{A^{(i)}_2 x} > 1$. It follows that $d \ge d A = k \sum_{i=1}^{d}\dist{A^{(i)} x}_2^2 > d$.
This contradiction ends the proof of the first property.

For any $x \in \Rspace^d$, such that $\dist{x}_2 > c$, the probability:

\[\prob{\dist{A^{(i)}x}_2 \le \alpha} =
  \prob{\frac{\dist{A^{(i)}x}_2^2}{c^2} \le (\frac{\alpha}{c})^2} \le
  \prob{\frac{\dist{A^{(i)}x}_2^2}{\dist{x}_2^2} \le (\frac{\alpha}{c})^2}. \]
  
Applying Lemma \ref{jl1} ends the proof.
\end{proof}

The algorithm works as follows: for each $i$, we project $\Rspace^d$ to
$\Rspace^k$ using $A_i$ and solve the corresponding problem in the smaller
space. For each query point, we need to merge the solutions obtained for each
sub-problem.  This results in reducing the $\CNN{c}{d}$ to $d/k$ instances of
$\CNN{\alpha}{k}$.

\begin{lemma}[Lemma 2 in \cite{wygos_red}]\label{lemalgo}
For $1 < \alpha < c$ and $k < d$, the $\CNN{c}{d}$ can be reduced to $d/k$ instances of the $\CNN{\alpha}{k}$.
The expected pre-processing time equals $\Oc(d^{\omega-1} n + d/k\ \pre{\alpha}{k})$ and
the expected query time equals $\Oc(d^{\omega-1} + d/k\ e^{\frac{k}{2}\big(1 - (\frac{\alpha}{c})^2 + \log((\frac{\alpha}{c})^2)\big)} n + d/k\ \que{k}{\alpha})$.
\end{lemma}
\begin{proof}
We use the assumption that $k < d$ and $d^{4-\omega} < n$ to simplify the complexities.
The pre-processing time consists of:
\begin{itemize}
    \item $d^3$: the time of computing a random orthonormal basis of $\Rspace^d$.
    \item $d^{\omega-1} n$: the time of changing the basis to $a_1, a_2, \dots, a_d$.
    \item $d n k$: the time of computing $A^{(i)} x$ for all $1 \le i \le d$ and for all $n$ points.
    \item $d/k\ preproc(\alpha, k)$: the expected pre-processing time of all sub-problems.

\end{itemize}
The query time consists of:
\begin{itemize}
    \item $d^{\omega-1}$: the amortized time of changing the basis to $a_1, a_2, \dots, a_d$.
    \item $d/k\ e^{\frac{k}{2}\big(1 - (\frac{\alpha}{c})^2 + \log((\frac{\alpha}{c})^2)\big)} n$: the expected number of false positives (by Lemma \ref{rjl}).
    \item $d/k\ query(k, \alpha)$: the expected query time of all sub-problems.
\end{itemize}
\end{proof}

The following corollary simplifies the formulas used in Lemma \ref{lemalgo} and
shows that  the $\CNN{c}{d}$ can be reduced
to a number of problems of dimension $\log{n}$ in an efficient way.  Namely, setting
$k=\big(\frac{2(1-\nu)}{1 - (\frac{\alpha}{c})^2 + \log((\frac{\alpha}{c})^2)}\big)\log{n}$ we get (presented also  in related work):

\primecor*

\end{document}